\documentclass[11pt]{article}
\usepackage[T1]{fontenc}

\usepackage[a4paper,top=3cm,bottom=2cm,left=3cm,right=3cm,marginparwidth=1.75cm]{geometry}

\usepackage[english]{babel}
\usepackage[utf8x]{inputenc}
\usepackage{amsmath}
\usepackage{graphicx}
\usepackage[colorinlistoftodos]{todonotes}
\usepackage{hyperref} 
\usepackage{color}
\usepackage{amssymb}
\usepackage[caption=false]{subfig}
\hypersetup{
    colorlinks,
    citecolor=blue,
    filecolor=blue,
    linkcolor=blue,
    urlcolor=blue
}
\usepackage{multirow}

\usepackage{adjustbox}
\usepackage{makecell}

\usepackage{natbib}
\setcitestyle{authoryear,aysep={,}}
\usepackage{amsthm}

\newtheorem{thm}{Theorem}
\newtheorem{lemma}{Lemma}
\newtheorem{definition}{Definition}
\newtheorem{cor}{Corollary}

\begin{document}
\title{Identifying Heritable Communities of Microbiome by Root-Unifrac and Wishart Distribution}
\date{}
\author{Yunfan Tang and Dan L. Nicolae \\
{\normalsize\emph{University of Chicago }}}

\maketitle

\begin{abstract}
We introduce a method to identify heritable microbiome communities when the input is a pairwise dissimilarity matrix among all samples. Current methods target each taxon individually and are unable to take advantage of their phylogenetic relationships. In contrast, our approach focuses on community heritability by using the root-Unifrac to summarize the microbiome samples through their pairwise dissimilarities while taking the phylogeny into account. The resulting dissimilarity matrix is then transformed into an outer product matrix and further modeled through a Wishart distribution with the same set of variance components as in the univariate model. Directly modeling the entire dissimilarity matrix allows us to bypass any dimension reduction steps. An important contribution of our work is to prove the positive definiteness of such outer product matrix, hence the applicability of the Wishart distribution. Simulation shows that this community heritability approach has higher power than existing methods to identify heritable groups of taxa. Empirical results on the TwinsUK dataset are also provided. 


\end{abstract}

\section{Introduction}
The human microbiome refers to the entire collection of microorganisms living inside the human host. Recently, a number of studies have established the association of microbiome, especially gut microbiome, with our metabolic and immune system \citep{cho2012human,gevers2012human,huttenhower2012structure}. Microbiome is known to be shaped by environmental factors such as age, hygiene and life style. On the other hand, genetic variation can lead to differences in food preferences, enzyme activity or immune response, hence having a nontrivial impact on microbial compositions. Therefore, it is of great scientific interest to identify taxa with overall variation significantly contributed from genetics. The advent of 16s rRNA sequencing platforms has led to much lower sequencing cost with increased resolution. This enables researchers to carry out large cohort studies that has discovered many heritable taxa \citep{goodrich2014human, goodrich2016genetic}. 

Most of the current heritability studies are based on variance component models. Given a univariate trait and its variance decomposition over genetic and environment components, heritability is defined as the proportion of total variance that can be explained by the genetic variance. One widely used variance component model is the Additive Genetics, Common Environment, Unique Environment (ACE) model \citep{eaves1978model}. The ACE variance component model has long been used in familial studies and recently applied to unrelated individuals  \citep{yang2011gcta} for a number of univariate traits such as height \citep{yang2010common}, inflammatory bowel diseases \citep{chen2014estimation} and diabetes \citep{bonnefond2015rare}. Several studies to identify heritable group of bacterial taxa \citep{goodrich2014human,davenport2015genome,goodrich2016genetic} follow this line of work by using the sum of the relative abundance over a group of taxa, after proper transformation, as the univariate response. There are also studies that extend the variance component model to multivariate case such as finding the linear combination of traits that maximizes heritability \citep{oualkacha2012principal} or using weighted average of heritabilities of each individual trait \citep{ge2016multidimensional}. 

Unfortunately, none of the aforementioned methods is designed for one of the most unique characteristics in microbial data, which is the phylogenetic tree among all taxa. This phylogenetic tree can be used to construct an ecologically meaningful way to summarize dissimilarities between each pair of microbial communities. Such dissimilarity, also called beta diversity, is sensitive to the evolutionary relatedness among the set of taxa that are responsible for the overall variation, and can potentially lead to increased power to detect heritable communities. Its focus on community level rather than individual taxon is also recommended in \cite{van2015rethinking}, which argues that it better captures the interaction of the entire microbiome with the human host and leads to more reasonable interpretation of heritability. Consequently, we aim to design a heritability model that takes these pairwise dissimilarities as input, essentially answering whether genetically similar subjects carry phylogenetically similar microbial communities.

The focus of this paper is on detection of heritable communities rather than on quantifying the precise amount of variability due to genetic similarity. We believe that quantification is less relevant for microbiome phenotypes than for other human traits because the environmental factors affecting composition vary dramatically even within racially homogeneous populations. In addition, different taxa have different susceptibility to environmental and genetic effects, but there has been no consensus on the right way to calculate average heritability on a community of taxa. Results from powerful methods for detecting heritable communities can be used for designing efficient follow-up studies for investigating the molecular mechanisms of host-microbiome interactions. This is the reason for emphasizing power in our simulation studies.

Estimating community heritability based on beta diversity has a major difficulty in statistical modeling since the response variable is a matrix of pairwise dissimilarities. To comply with the traditional heritability models, ordination methods such as non-metric multidimensional scaling (NMDS) and principal coordinate analysis (PCoA) must be applied to find the univariate representation that best preserves the original pairwise distance. Obviously, different ordination standards can lead to different heritability results. In addition, the recovered univariate response usually represents only a fraction of total variation in the dissimilarity matrix and has unclear biological meanings. These difficulties altogether point to the necessity of a statistical model capable of decomposing total variation in a dissimilarity matrix without any transformation.

A crucial property of a dissimilarity or distance matrix, as pointed out by \cite{gower1966some}, is that it can be transformed into an outer product matrix. This outer product matrix can be conveniently modeled by the Wishart distribution, which has a straightforward analogy to the univariate ACE model, hence definition of heritability, by imposing a similar additive form on the covariance matrix parameter. However, Wishart distribution is only applicable when the response is a positive definite matrix, a requirement not satisfied by most beta diversities. A major contribution of this paper is that we prove this property for a particular beta diversity, the square root transformation of weighted Unifrac. Unifrac \citep{lozupone2005unifrac, lozupone2007quantitative} incorporates phylogenetic information among bacterial species and has been extensively applied to a number of microbiome studies. To the best of our knowledge, no prior work exists to model the entire variation in Unifrac matrix for heritability analysis. 

The rest of this paper is organized as follows. Section \ref{sec:Wishart} introduces the Wishart model with ACE variance components. Section \ref{sec:community heritability} proves that the square root of weighted Unifrac is applicable for the Wishart distribution. Section \ref{sec:simulation} compares the power of detecting non-zero heritability between our method and other current methods in simulation. Section \ref{sec:TwinsUK} provides empirical results using TwinsUK fecal microbiome data. Section \ref{sec:discussion} concludes this paper with further discussions.

\section{Wishart distribution with variance components} \label{sec:Wishart}
We start from reviewing the ACE variance component model (A for additive genetics, C for common environment and E for unique environment) for heritability analysis on univariate traits \citep{eaves1978model}. This model assumes additive random effects from genetic factors, common environments and unique environments. Let $n$ be the number of samples, $\boldsymbol y$ be an $n \times 1$ vector of their univariate traits, $\boldsymbol X$ be a $n \times m$ matrix of covariates such as the intercept, age, sex and weight, and $\boldsymbol \beta$ be an $m \times 1$ vector of fixed effects. Furthermore, let $\boldsymbol A$ be an $n \times n$ genetic relationship matrix (GRM), $\boldsymbol C$ be an $n \times n$ matrix that quantifies shared environments, and $\boldsymbol E = \boldsymbol I_n$ be an $n \times n$ identity matrix for the unique environment effects. The GRM quantifies additive genetic covariance among individuals. In familial studies, $\boldsymbol A$ is twice the kinship matrix. For example, $\boldsymbol A_{i,j} = 1$ for monozygotic twins and $\boldsymbol A_{i,j} = 1/2$ for dizygotic twins. Furthermore, $\boldsymbol C_{i,j} = 1$ if and only if $i$th and $j$th individual share the same household. The diagonal entries of $\boldsymbol A$ and $\boldsymbol C$ are all set to one. In genome wide association studies on unrelated individuals, $\boldsymbol A$ can be estimated by SNP data \citep{yang2011gcta} and the shared environment matrix is usually omitted. 

The ACE variance component model takes the following form:
\begin{equation} \label{ACE}
\boldsymbol y = \boldsymbol X \boldsymbol\beta + \boldsymbol g + \boldsymbol c + \boldsymbol e, \hspace{4mm} \boldsymbol g \sim N(0, \sigma^2_A \boldsymbol A), \hspace{4mm} \boldsymbol c \sim N(0, \sigma^2_C \boldsymbol C), \hspace{4mm}  \boldsymbol e \sim N(0, \sigma^2_E \boldsymbol E) 
\end{equation}
where $\boldsymbol g$, $\boldsymbol c$ and $\boldsymbol e$ are assumed to be mutually independent.

Heritability ($h$) is defined as the proportion of total variance that is due to genetic factors:
\begin{equation} \label{heritability}
h = \frac{\sigma^2_A}{\sigma^2_A + \sigma^2_C + \sigma^2_E}
\end{equation}

A common approach to estimate $\sigma^2 = (\sigma^2_A, \sigma^2_C, \sigma^2_E)$ and hence $h$ uses residual maximum likelihood (REML) \citep{yang2011gcta,zhou2012genome}. Let $\boldsymbol L$ be the $(n-m) \times n$ matrix with its rows spanning the kernel space of $\boldsymbol X'$. Left multiplying (\ref{ACE}) by $\boldsymbol L$ leads to $\boldsymbol{Ly} \sim N(\boldsymbol 0, \boldsymbol{L\Sigma L'})$ where $\boldsymbol\Sigma =  \sigma^2_A \boldsymbol A + \sigma^2_C \boldsymbol C + \sigma^2_E \boldsymbol E$, after which one maximizes its likelihood to obtain REML estimates $\hat \sigma^2$ and $\hat h$. The REML likelihood takes the following form:
\begin{equation} \label{normal_REML}
l(\sigma^2; \boldsymbol y) = -\frac{n-m}{2}\log(2\pi) - \frac{1}{2}\log|\boldsymbol{L\Sigma L'}| - \frac{1}{2} \boldsymbol{y'L'\Sigma^{-1}Ly}
\end{equation}

We are interested in extending the ACE framework to the case where we can only observe an outer product matrix or covariance matrix, instead of the raw values of the univariate traits. To start, notice that the ACE model implies that
\begin{equation} \label{expectation_outer}
E(\boldsymbol {Lyy}'\boldsymbol L') = \boldsymbol L \boldsymbol \Sigma \boldsymbol L'
\end{equation}
where $\boldsymbol{yy'}$ serves as a sample outer product matrix. Now suppose we can only observe an outer product matrix $\boldsymbol M$ but not $\boldsymbol y$. This happens when one analyze a dataset by measuring its pairwise dissimilarities and apply principal coordinate analysis (details provided in the next section). Since both $\boldsymbol{yy}'$ and $\boldsymbol M$ have the same interpretation, an analogy to (\ref{expectation_outer}) would be
\begin{equation} \label{expectation_outer_M}
E(\boldsymbol {LML'}) = \boldsymbol L \boldsymbol \Sigma \boldsymbol L'
\end{equation}
A similar analogy is used by \cite{mcardle2001fitting} to derive a pseudo F-statistic to test fixed effects when the observation is a pairwise dissimilarity matrix. The effect of $\boldsymbol L$ is, similar to the univariate case, to remove the fixed effects in the outer product matrix $\boldsymbol M$. One can interpret this by expanding $\boldsymbol M$ into the sum of rank 1 matrices: $\boldsymbol M = \sum_{i=1}^{\text{rank}(M)} \boldsymbol M_{i} \boldsymbol M_{i}' $ with $\boldsymbol M_{i} \in \mathbb R^n$, and imposing that $E(\boldsymbol M_{i}) = \boldsymbol{X\beta}_i $.

In particular, (\ref{expectation_outer_M}) suggests that we can use a Wishart distribution to model $\boldsymbol Z =  \boldsymbol{LML}'$: $\boldsymbol Z \sim W(\boldsymbol{L\Sigma L'}/q, q)$ in order to align with the form of expectation in (\ref{expectation_outer_M}). Without affecting heritability estimates, we can further remove $q$ from the scale matrix and simply write $\boldsymbol Z \sim W(\boldsymbol{L\Sigma L'}, q)$, which gives the following log likelihood: 

\begin{align} 
l(q, \sigma^2; \boldsymbol Z) &= -\frac{q}{2}\log|\boldsymbol{L\Sigma L'}| - \frac{1}{2}\text{tr}\big((\boldsymbol{L\Sigma L'})^{-1} \boldsymbol Z\big) + \frac{q-(n-m)-1}{2}\log|\boldsymbol Z|  \nonumber \\
& \hspace{4.5mm} - \frac{q(n-m)}{2}\log 2 - \log \Gamma_{n-m}(\frac{q}{2})  \label{Wishart_REML}
\end{align}
where $\Gamma_{n-m}(\cdot)$ is the multivariate gamma function and $q$ can be any real number larger than $n-m-1$. Maximizing (\ref{Wishart_REML}) leads to the MLEs $(\hat q, \hat \sigma^2)$ and hence $\hat h$ from (\ref{heritability}) by using $\hat \sigma^2$. The gradient of (\ref{Wishart_REML}) with respect to $\sigma^2$ is very similar to the case of (\ref{normal_REML}), and the partial derivate of $q$ is straightforward to obtain. We use gradient based optimization, such as L-BFGS \citep{liu1989limited}, to obtain the MLEs $(\hat q, \hat \sigma^2)$.

The log likelihood (\ref{Wishart_REML}) is only applicable when $\boldsymbol Z$ is positive definite. The next section will prove this condition when $\boldsymbol Z$ is calculated from a particular microbiome beta diversity metric.

\section{Community heritability by root-Unifrac and Wishart distribution}\label{sec:community heritability}
Unifrac \citep{lozupone2005unifrac,lozupone2007quantitative} is one of the most popular metrics to quantify pairwise dissimilarities among microbial communities. A common way to incorporate Unifrac into the ACE model (\ref{ACE}) is to apply principal coordinate analysis (PCoA) on the $n \times n$ Unifrac dissimilarity matrix and then use each of the principal eigenvectors separately as a univariate response  \citep{goodrich2014human,goodrich2016genetic,quigley2017heritability}. Specifically, let $u(i,j)$ be the Unifrac dissimilarity between $i$th and $j$th sample, and $\boldsymbol D$ be an $n \times n$ matrix satisfying $\boldsymbol D_{i,j} = -u(i,j)^2/2$. Also, define $\boldsymbol J = \boldsymbol I_n - \boldsymbol 1_n \boldsymbol 1_n'/n$ where $\boldsymbol 1_n$ is a unit vector of length $n$. PCoA first calculates Gower's centered matrix as $\boldsymbol M = \boldsymbol {JDJ}$ \citep{gower1966some} , which turns a dissimilarity matrix into a centered outer product matrix. This is because if $u(i,j)$ happens to be the Euclidean distance between the pair of vectors $\boldsymbol c_i$ and $\boldsymbol c_j$ for $1\leq i \leq n$, then it is easy to deduce that
\begin{equation*}
\boldsymbol M = \boldsymbol {JCC'J} , \hspace{4mm} \boldsymbol C = \begin{pmatrix}
\boldsymbol c_1' \\
\boldsymbol c_2' \\
... \\
\boldsymbol c_n' \\
\end{pmatrix}
\end{equation*}
After this step, one applies eigenvalue decomposition on $\boldsymbol M$ to obtain the principal eigenvectors. Each of the principal eigenvectors is separately used as a univariate trait in (\ref{normal_REML}).

Although the principal eigenvectors of $\boldsymbol M$ can quantify community information to some extent, they are hard to interpret and usually express only a fraction of the total variation in the Unifrac matrix. An alternative is to directly use $\boldsymbol M$ as the observation, which has been applied to nonparametric testing of fixed effects \citep{mcardle2001fitting} and  used in a number of microbial studies \citep{chen2012associating,wang2016genome}. Given the nature of $\boldsymbol M$ being an outer product matrix, it is reasonable to model $\boldsymbol{Z = LML'}$ as generated from a Wishart distribution, allowing us to maximize (\ref{Wishart_REML}) to obtain the heritability estimate. However, the central difficulty is that $\boldsymbol{Z}$ may not be positive definite and therefore its log determinant in (\ref{Wishart_REML}) can be undefined.

In this section, we present our result stating that using the square root transformation of the weighted Unifrac \citep{lozupone2007quantitative}, which we call root-Unifrac, will guarantee that $\boldsymbol{LML'}$ is positive definite under a mild condition. For 16S rRNA data clustered into operational taxonomic units (OTUs), suppose there are $R$ OTUs in total and let $\boldsymbol x_i = (x_{i1}, x_{i2}, ..., x_{iR})$ denote the number of sequences belonging to each of the $R$ OTUs in the $i$th microbial sample. The OTU relative abundance is calculated as $\boldsymbol \theta_i = \boldsymbol x_i / \sum_{r=1}^{R} x_{ir}$ where $\boldsymbol \theta_i = (\theta_{i1}, ..., \theta_{iR})$. Now suppose that we have a rooted phylogenetic tree with $K$ branches. Let $b_k$ be the length of $k$th branch and 
$p_{i,k}$ be the sum of $\theta_{ir}$ over all $r$'s that are under the $k$th branch.
The root-Unifrac is defined as
\begin{equation}\label{root-Unifrac}
u(i,j) = \sqrt{\sum_{k=1}^K b_k |p_{i,k} - p_{j,k} |}
\end{equation}
Similar to the original Unifrac, the root-Unifrac takes phylogenetic information and account for taxa relatedness while comparing different communities. It is simple to show that the root-Unifrac satisfies non-negativity, symmetry and triangle inequality. Therefore, we can define a finite metric space $(\Omega, u)$ where $\Omega = (\omega_1,\omega_2,...,\omega_n)$ correspond to the $n$ microbial samples and $u(\omega_i, \omega_j) = u(i,j)$ according to (\ref{root-Unifrac}). We prove the positive definiteness of $\boldsymbol{LML}'$ by showing that $(\Omega,u)$ has an isometric embedding into an Euclidean space with dimension at least $n-1$, as long as there exists a $k^*$ such that $\{p_{i,k^*}\}_i$ are all different. 

Our main results are the following:
\begin{thm} \label{thm1}
Let $u(i,j) = \sqrt{\sum_{k=1}^K b_k |p_{i,k} - p_{j,k}|}$ and define an $n\times n$ matrix $\boldsymbol D$ satisfying $\boldsymbol D_{i,j} = -u(i,j)^2/2$. Suppose that $\exists k^* \in \{1,2,...,K\}$ such that $\{p_{i,k^*}\}_i$ are all different. Then the Gower's centered matrix, $\boldsymbol{M=JDJ}$, is positive semidefinite with rank $n-1$.
\end{thm}

\begin{cor} \label{cor1}
Assume that the covariate matrix $\boldsymbol X$ includes the intercept. Under the existence of $k^*$ in Theorem \ref{thm1}, $\boldsymbol{LML'}$ is positive definite. 
\end{cor}

We present the proofs of Theorem \ref{thm1} and Corollary \ref{cor1} in the Appendix. Corollary \ref{cor1} guarantees the applicability of Wishart likelihood to $\boldsymbol{Z} = \boldsymbol{LML}'$. In addition, Theorem \ref{thm1} shows that there will be no negative eigenvalues in $\boldsymbol M$, hence no imaginary coordinates present in PCoA.

Our proof of positive definiteness only uses the fact that $u(i,j)$ is the square root of sum of absolute values. Therefore, it is still applicable when the summation in (\ref{root-Unifrac}) is invoked over only a subset of branches. This is useful when we are concerned with a portion of community difference that comes from a particular subset of OTUs. For example, this subset can be chosen according to either a particular taxon such as Firmicutes, or from a certain internal node on the phylogenetic tree. Formally, let $\mathcal R$ be the subset of OTUs of interest and $k(\mathcal R)$ be the set of branches such that each branch included has all of its children OTU belonging to $\mathcal R$. The root-Unifrac distance contributed by OTUs from $\mathcal R$ is defined as
\begin{equation} \label{root-Unifrac-taxa}
u_{\mathcal R}(i,j) = \sqrt{\sum_{k \in k(\mathcal R)} b_k |p_{i,k} - p_{j,k}|}
\end{equation}
Using Corollary \ref{cor1} and substituting $\boldsymbol{Z = LML'}$ into the Wishart variance component model (\ref{Wishart_REML}), we can obtain the a heritability estimate $\hat h_{\mathcal R}$ for each taxa.

We use permutation to test $H_0: h = 0$ versus $H_a: h > 0$ and resampling method to obtain the confidence interval of heritability. P-value is calculated as ${|\{i: \hat h^{(i)} > \hat h\}|} / {n_{\text{perm}}}$, where $\hat h^{(i)}$ is the $i$th round heritability estimate after permuting the rows and columns of the GRM matrix $\boldsymbol A$ and $n_\text{perm}$ is the total rounds of permutation. For confidence interval, notice that observations across different families are assumed independent in familial studies. Therefore, we resample (bootstrap) families with replacement while keeping total number of families the same for a total of $n_\text{boot}$ rounds. Heritability estimate in each round is obtained by using all resampled observations followed by constructing $\boldsymbol A$ and $\boldsymbol C$ to preserve the same intra-family covariances and inter-family independences. Let $\tilde{\boldsymbol h}$ be the vector of $n_{\text{boot}}$ heritability estimates from the resampling procedure above. Then the confidence interval at level $1-\alpha$ is constructed as $\big(\hat h + z_{\alpha/2} \text{se}(\tilde{\boldsymbol h}), \hat h + z_{1-\alpha/2} 
\text{se}(\tilde{\boldsymbol h}) \big) $, where $\text{se}(\tilde{\boldsymbol h})$ is the sample standard error of $\tilde{\boldsymbol h}$ and $z_\alpha$ is the $\alpha$ quantile of normal distribution.

\section{Power simulation} \label{sec:simulation}
As we mentioned in the introduction, an obvious advantage to use beta diversity to summarize microbiome data is facilitating biological interpretation of heritability. An equally, if not more, important concern is whether such method with root-Unifrac has improved power to identify heritable groups of taxa. This section compares the simulated powers of detecting non-zero heritability between our model and other current methods. We first simulate OTU absolute abundances $\boldsymbol a_i = (a_{i1},...,a_{iR})$ according to log normal distribution: $\log \boldsymbol a_i \sim N(\mu_0, \Sigma_0)$ where $\Sigma_0 = \sigma_0^2 \big((1-\rho_0)\boldsymbol I_R+\rho_0 \boldsymbol 1_R \boldsymbol 1'_R\big)$ for $i=1,2,...,n$. Normalizing the absolute abundance gives the relative abundance $\boldsymbol \theta_i = \boldsymbol a_i / \sum_{r=1}^R a_{ir}$ for each $i$. In this section, we set $R=10$, $\mu_0 = 0$ and $\sigma^2_0 = 2$. Each round of simulation has $n=200$ samples including 50 MZ twin pairs and 50 DZ twin pairs. Similarities within each twin pair are generated by shrinking their relative abundances towards the geometric family mean. The shrinkage procedure is only invoked for the first 6 OTUs in order to produce localized signal. In other words, let $\gamma_0$ and $\gamma_1$, both within $[0,1]$, be the shrinkage strength for MZ and DZ twins, respectively. Larger shrinkage value corresponds to greater signal strength. For simplicity, we set $\gamma_1 = \gamma_0 /2$. The shrinkage procedure is invoked to each twin pair $(i,j)$ in the following way
$$\theta_{i'r} \leftarrow (\sqrt{\theta_{ir}\theta_{jr}})^{\gamma_z} \theta_{i'r}^{1-\gamma_z}  \text{ for $i'\in \{i,j\}$ and $1\leq r \leq 6$} $$
where $z=0$ if MZ twin pair and $z=1$ for DZ twin pair. After these two steps, each $\boldsymbol \theta_i$ is renormalized again to impose the unit sum constraint.

We use the phylogenetic tree of 10 OTUs belonging to the Rikenellaceae family in the TwinsUK dataset (to be introduced in the next section) to calculate root-Unifrac metrics. This tree is plotted in Figure \ref{fig:simulation_tree} with its internal nodes labeled as $T_1, T_2, ..., T_9$:
\begin{figure}[h] 
\centering
\includegraphics[height=7cm]{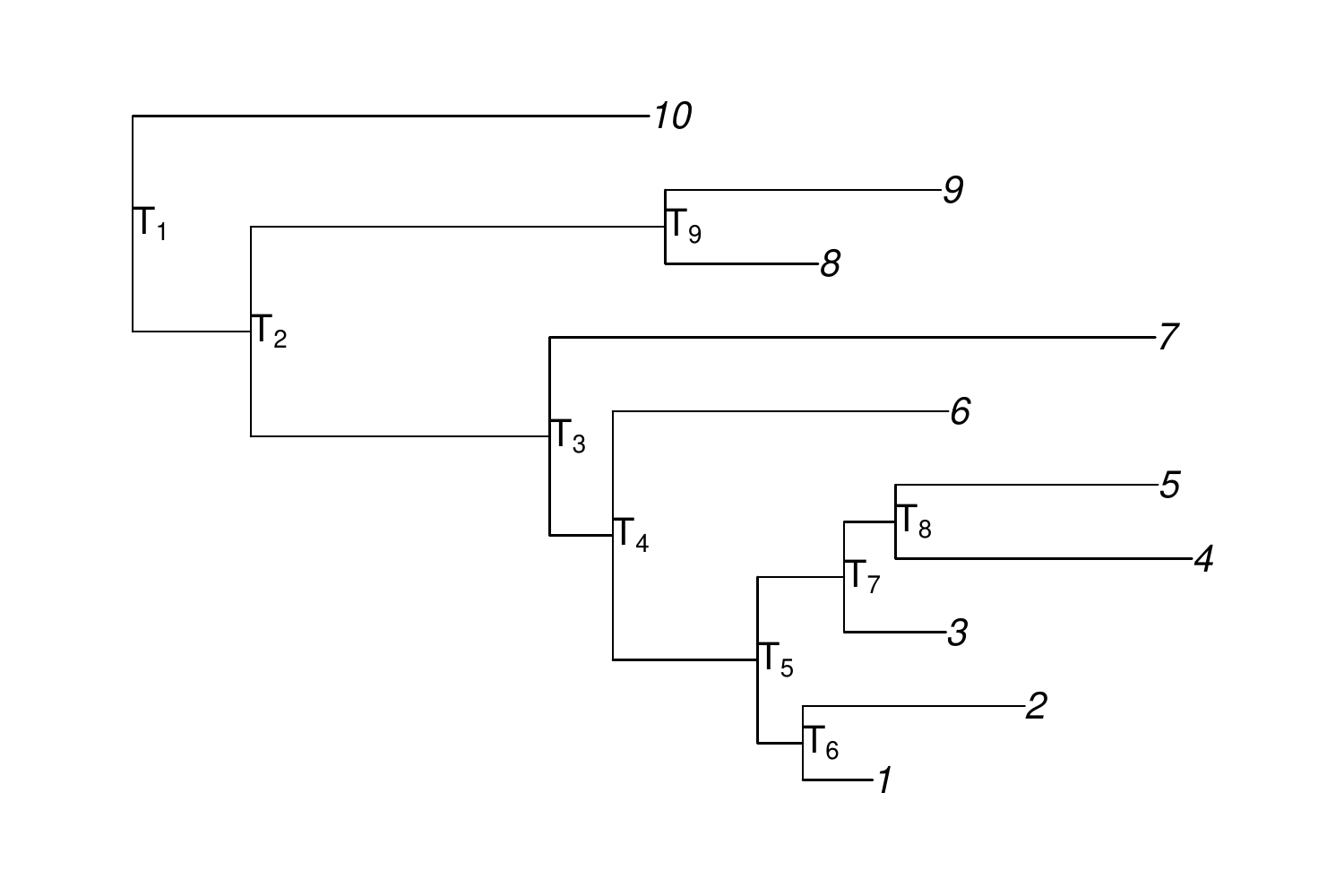}
\caption{Phylogenetic tree for power simulation.} \label{fig:simulation_tree}
\end{figure}

For each $T_e$, define $\mathcal R_e$ to be the set of OTUs under $T_e$. For example, $\mathcal R_5 = \{1,2,3,4,5\}$. Then we calculate heritability of each node $T_e$ from three different methods: 
\begin{enumerate}
\item Wishart: Use (\ref{root-Unifrac-taxa}) with $\mathcal R = \mathcal R_e$ to calculate $\boldsymbol Z$ and maximize (\ref{Wishart_REML}).
\item Univariate (logit): Let $\theta^{(e)}_i = \sum_{r \in \mathcal R_e} \theta_{ir}$. Maximize (\ref{normal_REML}) with $\boldsymbol y = (\log \frac{ \theta^{(e)}_1 }{ 1-\theta^{(e)}_1},..., \log \frac{\theta^{(e)}_n }{ 1-\theta^{(e)}_n})$.
\item Univariate (Box-Cox): Maximize (\ref{normal_REML}) with $\boldsymbol y$ being the optimal Box-Cox transformed response on $\theta^{(e)}_i$. This method is used in \cite{goodrich2014human} and \cite{goodrich2016genetic}.
\end{enumerate}

Notice that if $e=1$, then the (\ref{root-Unifrac-taxa}) is the same as (\ref{root-Unifrac}). The second and third method are the traditional univariate cases where, given $e$, it uses only the node relative abundance $\theta_i^{(e)}$ as input and completely ignores the relative contributions of relevant $\theta_{ir}$'s. This method cannot be applied to $e=1$ since $\theta^{(1)}_i = 1$ for all $i$. After calculating these heritability estimates, we permute the rows and columns of $\boldsymbol A$ for 100 times to obtain the p-value of testing the null hypothesis of zero heritability for each method. Finally, Type-I error or power are obtained by calculating the proportion of p-values below 0.05 within 200 rounds of simulation.

We compare the type I error and power among the aforementioned three methods as $\gamma_0$ varies while fixing $\rho_0 = 0$. The results are presented in Table \ref{table:powersim_gamma}. Type I error of all three methods at $\gamma_0 = 0$ are slightly less than, but not far from, the nominal level 0.05. For $T_9$, its power is also close to 0.05 at all values of $\gamma_0$ since this node does not contain any shrunk OTUs. When $\gamma_0 > 0$, we observe that Wishart method consistently has the highest power, compared to univariate (logit) and univariate (Box-Cox), on all nodes with at least one shrunk OTU, i.e. $\{T_1,.,,,T_8\}$. For a fixed value of $\gamma_0 > 0$, Wishart method has increased power for nodes that contain higher number of shrunk OTUs (first number in the parenthesis under the node column). Its test powers peaks at $T_4$, which includes all shrunk OTUs, i.e. $\{1,...,6\}$, and none of the unshrunk OTUs, i.e $\{7,...,10\}$. Surprisingly, both univariate (logit) and univariate (Box-Cox) method have decreasing power as number of shrunk OTUs increase under the node. Depending on the value of $\gamma_0$, their powers, as functions of the node, peak at either $T_6$ and $T_8$, both of which contain smallest non-zero number of shrunk OTUs. On other nodes such as $T_7$, $T_5$ and $T_4$, its power is much smaller than Wishart. This shows that methods using only the univariate trait $\{\theta_i^{(e)}\}_i$ are inadequate to detect community heritability even on a small group of OTUs. Table \ref{table:powersim_rho} presents a similar power comparison but fixing $\gamma_0 = 0.2$ and letting $\rho_0$ vary. Test powers of each method exhibit some fluctuations as $\rho_0$ assumes different values, but the overall conclusion remains the same. 

\renewcommand{\arraystretch}{1.2}
\setlength{\tabcolsep}{4pt}
\begin{table}[h]
\footnotesize
\caption{\label{table:powersim_gamma} Simulated type-I error and power from testing zero heritability for different values of $\gamma_0$ while $\rho_0 = 0$. Larger value of $\gamma_0$ indicate greater strength of signal. The number in the parenthesis after each $T_e$ denotes how many OTUs under $T_e$ are shrunk towards the family mean or not. For example, $T_1$ has 6 shrunk OTUs and 4 unshrunk OTUs. W denotes Wishart method, U(l) denotes univariate method with logit transform, and U(B) denotes univariate method with Box-Cox transform.}
\vspace{2mm}
\centering
\begin{tabular}{c c c c c c c c c c c c c c c c c}
\hline 
\rule{0pt}{2ex} & &  \multicolumn{3}{c}{$\gamma_0 = 0$ (null)} & & \multicolumn{3}{c}{$\gamma_0 = 0.1$} & & \multicolumn{3}{c}{$\gamma_0 = 0.2$} & & \multicolumn{3}{c}{$\gamma_0 = 0.3$}\\ \cline{3-5} \cline{7-9} \cline{11-13} \cline{15-17}   
\rule{0pt}{2ex} \textbf{Node} & & W & U(l) & U(B) & & W & U(l) & U(B) & & W & U(l) & U(B) & & W & U(l) & U(B) \\
\hline
$T_1$ (6,4) & & 0.02 & --- & --- & & 0.09 & --- & --- & & 0.185 & --- & --- & & 0.34 & --- & ---\\ 
$T_2$ (6,3) & & 0.035 & 0.05 & 0.05 & & 0.085 & 0.045 & 0.035 & & 0.21 & 0.035 & 0.04 & & 0.395 & 0.03 & 0.045 \\
$T_3$ (6,1) & & 0.04 & 0.015 & 0.03 & & 0.13 & 0.035 & 0.045 & & 0.295 & 0.025 & 0.035 & & 0.55 & 0.045 & 0.06 \\
$T_4$ (6,0) & & 0.04 & 0.03 & 0.04 & & 0.185 & 0.035 & 0.04 & & 0.555 & 0.065 & 0.06 & & 0.8 & 0.13 & 0.115  \\
$T_5$ (5,0) & & 0.02 & 0.035 & 0.05 & & 0.14 & 0.045 & 0.06 & & 0.475 & 0.09 & 0.1 & & 0.765 & 0.185 & 0.16 \\
$T_6$ (2,0) & & 0.04 & 0.025 & 0.04 & & 0.125 & 0.085 & 0.08 & & 0.335 & 0.23 & 0.235 & & 0.585 & 0.445 & 0.41 \\
$T_7$ (3,0) & & 0.015 & 0.025 & 0.025 & & 0.12 & 0.065 & 0.065 & & 0.44 & 0.19 & 0.185 & & 0.695 & 0.38 & 0.38  \\
$T_8$ (2,0) & & 0.03 & 0.03 & 0.015 & & 0.1 & 0.04 & 0.05 & & 0.37 & 0.22 & 0.205 & & 0.68 & 0.49 & 0.46 \\
$T_9$ (0,2) & & 0.03 & 0.025 & 0.025 & & 0.03 & 0.03 & 0.025 & & 0.04 & 0.04 & 0.035 & & 0.04 & 0.035 & 0.04\\
\hline
\end{tabular}
\end{table}
\renewcommand{\arraystretch}{1}
\setlength{\tabcolsep}{6pt}

\renewcommand{\arraystretch}{1.2}
\begin{table}[h]
\footnotesize
\caption{\label{table:powersim_rho} Simulated type-I error and power from testing zero heritability for different values of $\rho_0$ while $\gamma_0 = 0.2$. The number in the parenthesis after each $T_e$ denotes how many OTUs under $T_e$ are shrunk towards the family mean or not. For example, $T_1$ has 6 shrunk OTUs and 4 unshrunk OTUs. W denotes Wishart method, U(l) denotes univariate method with logit transform, and U(B) denotes univariate method with Box-Cox transform.}
\vspace{2mm}
\centering
\begin{tabular}{c c c c c c c c c c c c c}
\hline 
\rule{0pt}{2ex} & &  \multicolumn{3}{c}{$\rho_0 = 0$} & & \multicolumn{3}{c}{$\rho_0 = 0.3$} & & \multicolumn{3}{c}{$\rho_0 = 0.6$} \\ \cline{3-5} \cline{7-9} \cline{11-13}
\rule{0pt}{2ex} \textbf{Node} & & W & U(l) & U(B) & & W & U(l) & U(B)  & & W & U(l) & U(B)\\
\hline
$T_1$ (6,4) & & 0.185 & --- & --- & & 0.15 & --- & --- & & 0.21 & --- & ---  \\ 
$T_2$ (6,3) & & 0.21 & 0.035 & 0.04 & & 0.2 & 0.01 & 0.01 & & 0.275 & 0.03 & 0.045 \\
$T_3$ (6,1) & & 0.295 & 0.025 & 0.035 & & 0.3 & 0.025 & 0.025 & & 0.37 & 0.06 & 0.06 \\
$T_4$ (6,0) & & 0.555 & 0.065 & 0.06 & & 0.495 & 0.085 & 0.105 & & 0.61 & 0.095 & 0.115\\
$T_5$ (5,0) & & 0.475 & 0.09 & 0.1 & & 0.43 & 0.16 & 0.14 & & 0.545 & 0.1 & 0.1 \\
$T_6$ (2,0) & & 0.335 & 0.23 & 0.235 & & 0.275 & 0.21 & 0.215 & & 0.34 & 0.235 & 0.245 \\
$T_7$ (3,0) & & 0.44 & 0.19 & 0.185 & & 0.39 & 0.185 & 0.175 & & 0.49 & 0.195 & 0.21\\
$T_8$ (2,0) & & 0.37 & 0.22 & 0.205 & & 0.3 & 0.205 & 0.195 & & 0.455 & 0.245 & 0.255\\
$T_9$ (0,2) & & 0.04 & 0.04 & 0.035 & & 0.02 & 0.04 & 0.045 & & 0.03 & 0.025 & 0.02\\
\hline
\end{tabular}
\end{table}
\renewcommand{\arraystretch}{1}

Next, we apply common dimension reduction techniques to the root-Unifrac dissimilarity matrix (\ref{root-Unifrac-taxa}). These methods include principal coordinate analysis (PCoA), metric multidimensional scaling (mMDS), and non-metric multidimensional scaling (nMDS). PCoA finds the eigenvectors of the outer product matrix $\boldsymbol M$, whereas both mMDS and nMDS aim to minimize their particular stress functions \citep{borg2005modern} to approximate all pairwise dissimilarities. The following univariate traits are extracted from these dimension reduction methods: the top three principal coordinates (eigenvectors) from PCoA, the best one-dimensional representation from mMDS and the best one-dimensional representation from nMDS. Each univariate trait is fed into (\ref{normal_REML}) to obtain a heritability estimate and a permutation p-value. The results, as shown in Table \ref{table:powersim_MDS} from 100 rounds of simulation, demonstrate that all of the dimension reduction techniques have much reduced power compared to the Wishart method on $T_1,...,T_8$, nodes with at least one shrunk OTUs. We also do not see a consistent ranking of powers among the top three principal coordinates. Furthermore, nMDS has uncalibrated Type-I error on $T_9$.

\renewcommand{\arraystretch}{1.2}
\begin{table}[h]
\footnotesize
\caption{\label{table:powersim_MDS} Simulated type-I error and power for testing zero heritability from different dimension reduction techniques at $\rho_0 = 0$ and $\gamma_0 = 0.3$. The number in the parenthesis after each $T_e$ denotes how many OTUs under $T_e$ are shrunk towards the family mean or not. For example, $T_1$ has 6 shrunk OTUs and 4 unshrunk OTUs. PCoA, mMDS and mMDS are all applied to the root-Unifrac dissimilarity (\ref{root-Unifrac-taxa}) and separately used as univariate response into (\ref{normal_REML}). PC1-PC3 indicate the first, second or third principal coordinate (eigenvector) from PCoA, respectively. W denotes the Wishart method. }
\vspace{2mm}
\centering
\begin{tabular}{c c c c c c c c c c c c c c c c c}
\hline 
\rule{0pt}{3ex} \textbf{Node} & & W & PC1 & PC2 & PC3 & mMDS & nMDS \\
\hline
$T_1$ (6,4) & & 0.26 & 0.05 & 0.04 & 0.06 & 0.05 & 0.18\\ 
$T_2$ (6,3) & & 0.39 & 0.04 & 0.07 & 0.32 & 0.06 & 0.15\\
$T_3$ (6,1) & & 0.62 & 0.12 & 0.04 & 0.5 & 0.09 & 0.16 \\
$T_4$ (6,0) & & 0.88 & 0.27 & 0.46 & 0.43 & 0.11 & 0.18\\
$T_5$ (5,0) & & 0.82 & 0.28 & 0.55 & 0.45 & 0.13 & 0.17\\
$T_6$ (2,0) & & 0.59 & 0.48 & 0.34 & 0.08 & 0.15 & 0.22\\
$T_7$ (3,0) & & 0.76 & 0.41 & 0.48 & 0.36 & 0.16 & 0.14\\
$T_8$ (2,0) & & 0.64 & 0.51 & 0.4 & 0.07 & 0.12 & 0.22\\
$T_9$ (0,2) & & 0.05 & 0.05 & 0.03 & 0.03 & 0.02 & 0.22\\
\hline
\end{tabular}
\end{table}
\renewcommand{\arraystretch}{1}

\section{Empirical results from TwinsUK} \label{sec:TwinsUK}
\subsection{Heritability estimates} \label{sec:TwinsUK_sub1}
\cite{goodrich2014human} examined the influence of host genetics on fecal microbiome from a large twin-based study (TwinsUK). The TwinsUK population has more than 1000 16S rRNA microbial samples including 416 twin pairs. These sequences are processed by QIIME v1.9.1 \citep{caporaso2010qiime} to produce the OTUs at 97\% similarity level and the phylogenetic tree.  Samples with sequencing depth less than 10000 are discarded. We do not apply any rarefaction or subsampling before calculating the taxon abundances (this issue is further inspected in Section \ref{sec:TwinsUK_sub3}). Since the overwhelming majority of observations are from females (1061 females vs 20 males), we remove all male observations to avoid excessive variability on the sex effect. In the case of longitudinal observations for the same individual, only the first observation is used. This leaves 186 dizygotic (DZ) and 126 monozygotic (MZ) twin pairs. Similar to \cite{goodrich2014human}, OTUs that appear in less than 50\% of the microbial samples are excluded. The total number of remaining OTUs is 705. We also introduce a pseudo count in each OTU in all samples.


We apply the aforementioned ACE model with Wishart distribution on these microbial samples using the following covariates: age, body mass index, identity of technician (two), sequencing run (16 instrument runs) and shipment batch (8 shipments). These technical covariates are chosen according to \cite{goodrich2014human}. The root-Unifrac matrix is calculated by (\ref{root-Unifrac-taxa}) only for those taxa with at least 4 descendant OTUs. To eliminate the burden of multiple hypothesis testing, if a higher level taxon (e.g. phylum Firmicutes) has more than 95\% of its sequences belonging to one of its lower level taxon (e.g. order Clostridia), then the higher order taxon is excluded. This leaves a total of 26 taxa, each with its own root-Unifrac dissimilarity matrix and heritability estimate.

As described in Section \ref{sec:community heritability}, we permute the rows and columns of $\boldsymbol A$ for $10^{4}$ times to test the null hypothesis of non-zero heritability for each of these 26 taxa. A total of 6 taxa have p-values smaller than the Bonferroni threshold at 0.05 global Type-I error. Notice that this is a conservative correction due to the correlations among taxa abundances. We further calculate the 95\% bootstrap confidence interval for these significant taxa. These results are reported in Table \ref{table:Wishart}. The Bifidobacterium genus is also reported with significant heritability in \cite{goodrich2014human}, although these authors use the univariate (Box-Cox) method (c.f. Section \ref{sec:simulation}). Other significant taxa in \cite{goodrich2014human} that are related to our findings include Ruminococcaceae genus and Clostridiaceae family.

\renewcommand{\arraystretch}{1.2}
\begin{table}[h]
\caption{\label{table:Wishart} Number of OTUs, heritability estimates, p-values, 95\% bootstrap confidence intervals (CI) for taxa that are globally significant at 0.05 level under Bonferroni correction. Taxon names are provided at their finest (lowest) possible rank: kingdom (k), phylum (p), class (c), order (o), family (f) and genus (g). P-value and CI are computed using permutation and resampling, respectively.}
\vspace{2mm}
\footnotesize
\centering
\begin{tabular}{ c c c c c c c}
\hline
\textbf{\Gape[0.2cm][0.2cm]{Taxa}} & \# of OTUs & $\hat h$ & \textbf{P-value} & \textbf{CI} \\
\hline
Actinobacteria (p) & 10 & 0.223 & $< 10^{-4}$ & (0.091, 0.355)  \\
Clostridiales (o) & 590 & 0.110 & $<10^{-4}$ & (0.071, 0.149)\\
Christensenellaceae (f) & 7 & 0.185 & $10^{-4}$ & (0.075, 0.294) \\
Rikenellaceae (f) & 10 & 0.149 & $2\times 10^{-4}$ & (0.044, 0.254) \\
Ruminococcaceae (f) & 227 & 0.093  & $< 10^{-4}$ & (0.049, 0.137)\\
Bifidobacterium (g) & 5 & 0.231 & $<10^{-4}$ & (0.096, 0.366) \\
\hline
\end{tabular}
\end{table}
\renewcommand{\arraystretch}{1}

\subsection{Effect of sequencing noise} \label{sec:TwinsUK_sub3}
Calculating the Unifrac or root-Unifrac requires relative abundances as input. For each sample, these relative abundances are obtained by normalizing the $i$th taxa sequences $\boldsymbol x_i$ over their sum $N_i = \sum_{o=1}^{n_o} x_{io}$, the latter conventionally called library size or sequencing depth. This normalization step introduces an extra layer of data uncertainty that is not modeled by any of the variance components in the Wishart ACE model. As a result, estimates of $\sigma^2_A, \sigma^2_C$ and $\sigma^2_E$, hence heritability, can be biased. Larger sequencing depth will mostly likely lead to small sequencing variability and thus reduce the bias in the ACE variance component estimates. Although it is hard to deduce the closed form of this bias, the fact that such noises caused by normalization are independent across samples can more likely lead to an inflated $\hat{\sigma}^2_E$ and thus a downwards biased $\hat h$. 

Here we inspect the bias of heritability estimates caused by sequencing noise through simulation using the same TwinsUK dataset. For each sample $i$, we first calculate the observed relative abundance $\boldsymbol\theta_i = \boldsymbol x_i / N_i$. These observed relative abundances are treated as the ground truth relative abundances for simulation purpose. After this step, we obtain $\tilde{\boldsymbol x}_i$ by drawing from a multinomial distribution with probability $\boldsymbol \theta_i$ and total size (sequencing depth) $\xi$, where $\xi$ is 2500, 5000, 7500 or 10000. The $i$th simulated relative abundance is therefore $\tilde{\boldsymbol \theta}_i = \tilde{\boldsymbol x_i} / \xi$. In each simulation round, we calculate Wishart heritability estimates using $\{\tilde{\boldsymbol \theta}_i\}_i$ for the root-Unifrac metric on the six significant taxa reported in Table \ref{table:Wishart}. The ground truth heritability, on the other hand, is obtained by using $ \{\boldsymbol \theta_i\}_i$ to calculate the root-Unifrac metric. For each value of $\xi$, a total of 100 simulation rounds are conducted. We demonstrate the boxplot of these simulated heritability estimates and compare them against the ground truth heritability (dashed line) in Figure \ref{fig:seqnoise_simulation}. The negative bias of simulated heritability estimates is present in all cases, and they decrease to zero at increasing levels of $\xi$. At $\xi = 10000$, the simulated estimates are all very close to the ground truth, with error less than 0.01. Since the mean and standard deviation of actual sequencing depth in TwinsUK dataset is 56911 and 18461, respectively, we conclude that the negative biases on the heritability estimates reported in Table \ref{table:Wishart} are negligible.

\begin{figure}[t]
\centering
\subfloat{\includegraphics[height=4cm]{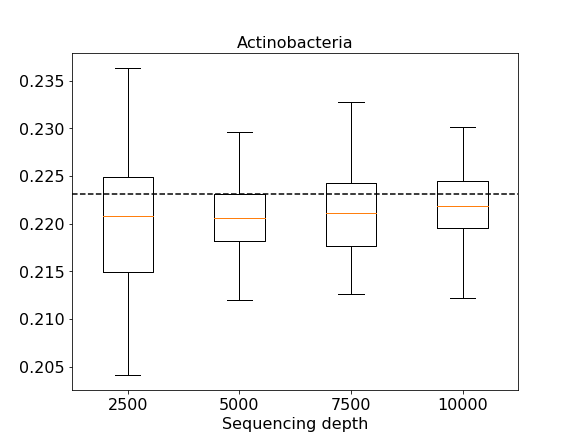}} 
\subfloat{\includegraphics[height=4cm]{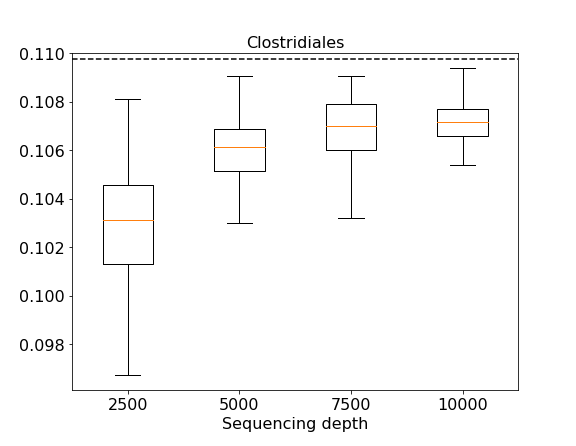}}
\subfloat{\includegraphics[height=4cm]{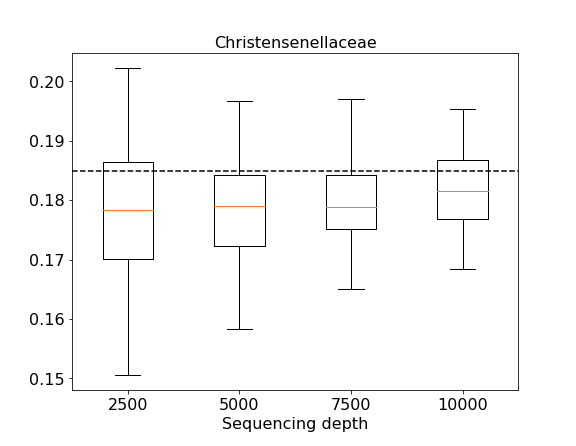}}  \\
\subfloat{\includegraphics[height=4cm]{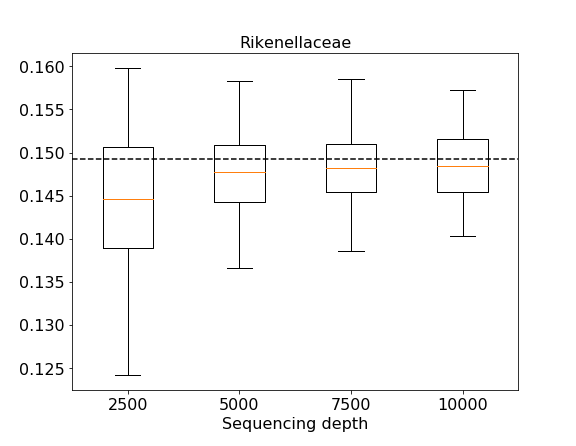}} 
\subfloat{\includegraphics[height=4cm]{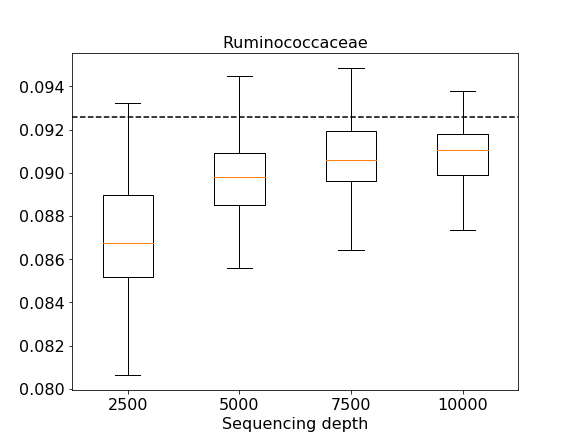}} 
\subfloat{\includegraphics[height=4cm]{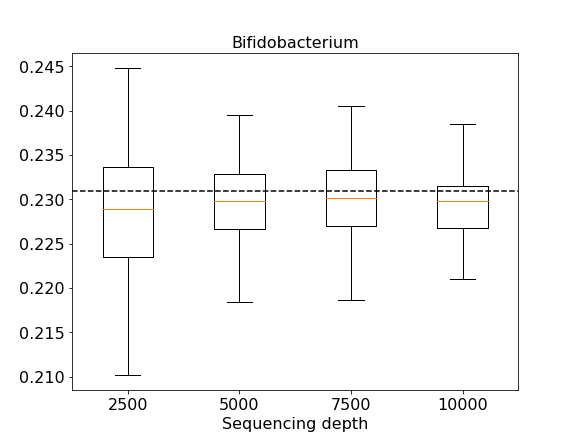}} \\
\caption{\label{fig:seqnoise_simulation} Boxplot of simulated heritability estimates from Wishart method. In each round, the $i$th sequencing data are produced by subsampling $\boldsymbol x_i$ to a certain sequencing depth $\xi$. A total of 100 simulation rounds are conducted for each value of $\xi \in \{2500, 5000, 7500, 10000\}$. Dashed lines in each plot correspond to the ground truth heritability calculated using $\{\boldsymbol\theta_i\}_i$ as input data.}
\end{figure}

\section{Discussion} \label{sec:discussion}

In this paper, we propose the Wishart variance component model to estimate microbiome community heritability when the microbiome data are summarized by their root-Unifrac dissimilarities. We prove that the root-Unifrac matrix always has an isometric Euclidean embedding and therefore is adequate for REML estimation with the Wishart distribution. Our work allows researchers to bypass the dimension reduction step and directly analyze all the variations present in the dissimilarity matrix.


In Section \ref{sec:TwinsUK_sub3} we inspected the negative biases of heritability estimates caused by sequencing noise. Although we concluded that such biases are negligible at the large sequencing depth of TwinsUK data, a better approach is to directly model the sequencing noise component as follows:
\begin{equation} \label{ACES}
\boldsymbol\Sigma =  \sigma^2_A \boldsymbol A + \sigma^2_C \boldsymbol C + \sigma^2_E \boldsymbol E + \boldsymbol S
\end{equation}
where $\boldsymbol S = \text{diag}(\sigma^2_{S_1}, \sigma^2_{S_2}, ..., \sigma^2_{S_n})$ captures sequencing noise for each sample, and heritability is still defined as $h = \sigma^2_A / (\sigma^2_A+\sigma^2_C+\sigma^2_E)$. This model makes it explicit that heritability is not dependent on sequencing noise.

Unfortunately, using (\ref{ACES}) leads to identifiability issues among $\sigma^2_{S_i}$'s and $\sigma^2_E$. One possible way to avert this obstacle is to separately estimate $\sigma^2_{S_i}$ by exploring the variability of sequences within $\boldsymbol x_i$. Suppose the $\boldsymbol \pi_i$ is the true relative abundance for $i$th individual. If we can find a reasonable distribution to model $\boldsymbol x_i | \boldsymbol \pi_i$, then we can generate independent and identically distributed samples, $\{\boldsymbol x^{[1]}_i, \boldsymbol x^{[2]}_i, ..., \boldsymbol x^{[B]}_i\}$, from such distribution by using $\hat{\boldsymbol \pi}_i = \boldsymbol x_i / N_i$ in order to mimic the process of repeatedly sequencing the $i$th sample for $B$ times. The sequencing depth of these bootstrap samples are kept at the same level at the original sample, i.e.  $\boldsymbol x^{[b]}_i = (x^{[b]}_{i1}, ..., x^{[b]}_{in_o})$ and $\sum_{u=1}^{n_o} x^{[b]}_{io} = \sum_{u=1}^{n_o} x_{io}$ for all $b$. Since $\{\boldsymbol x^{[b]}_i \}_b$ share the same effect from covariates, genetics, common environment and unique environment, we can use a single intercept to model their total effect. This leaves the independent and identical sequencing noise the only remaining component that explains their variability:
\begin{equation} \label{Wishart_i}
\boldsymbol L_1 \boldsymbol M_i \boldsymbol L_1' \sim W(\sigma^2_{S_i}\boldsymbol L_1 \boldsymbol L_1', q_i)
\end{equation}  
where $\boldsymbol M_i$ is the $B \times B$ Gower's centered matrix from calculating root-Unifrac on $\boldsymbol x^{[1]}_i, ..., \boldsymbol x^{[B]}_i$, and $\boldsymbol L_1$ is the $(B-1) \times B$ matrix that removes only the intercept effect. The estimated $\hat \sigma^2_{S_i}$ from maximizing the Wishart log likelihood of (\ref{Wishart_i}) can be used for (\ref{ACES}), hence avoiding the identifiability issue.

\bibliography{ref}

\appendix 
\section*{Appendix: Theorem proofs}
We first prove the the following lemma:
\begin{lemma} \label{lemma1}
For $a_1 > a_2 > ... > a_n > 0$, 
$$\det \begin{pmatrix}
a_1 & a_2 & a_3 & ... & a_n\\
a_2 & a_2 & a_3 & ... & a_n\\
a_3 & a_3 & a_3 & ... & a_n \\
... & ... \\
a_n & a_n & a_n & ... & a_n \\
\end{pmatrix} > 0$$
\end{lemma}
\begin{proof}
We prove by induction. Let $\boldsymbol C_i$ be the upper-left $i \times i$ corner of the matrix above. Evidently, $\det(\boldsymbol C_1) >0$ and $\det(\boldsymbol C_2) > 0$. 

Now assume that $\det(\boldsymbol C_{N-1})>0$ for some $N \geq 2$, we can write $\boldsymbol C_N$ as 
$$\boldsymbol C_N = \begin{pmatrix}
\boldsymbol C_{N-1} & a_N \boldsymbol 1_{N-1} \\
a_N \boldsymbol 1_{N-1}' & a_N \\
\end{pmatrix}
$$
Using the block formula for determinants, we have $\det(\boldsymbol C_N) = \det(\boldsymbol C_{N-1} - a_N \boldsymbol 1_{N-1} \boldsymbol 1_{N-1}') a_N$. Notice that $\boldsymbol C_{N-1} - a_N \boldsymbol 1_{N-1} \boldsymbol 1_{N-1}'$ also assumes the form of $\boldsymbol C_{N-1}$ except that $a_i$ is substituted by $a_i - a_N$ for $1\leq i \leq N-1$. Since $a_1-a_N > a_2-a_N > ... > a_{N-1}-a_N$, we know from the induction assumption that $\det(C_{N-1} - a_N \mathbf 1_{N-1} \mathbf 1_{N-1}')>0$, and therefore $\det(C_N) > 0$.
\end{proof}

We will need the following results from \cite{morgan1974embedding}:
\begin{definition} \label{def_Morgan}
\textnormal{\citep{morgan1974embedding}} Consider an ordered tuple $(\kappa_0, \kappa_1, ..., \kappa_{N})$ whose elements are from a metric space $(\Omega,d)$. Define an $N \times N$ matrix $\boldsymbol V$ such that $\boldsymbol V_{i,j} = \big(d^2(\kappa_i, \kappa_0)+d^2(\kappa_j, \kappa_0)-d^2(\kappa_i, \kappa_j)\big)/2$. $(\Omega, d)$ is called flat if $|\boldsymbol V|\geq 0$ for any ordered tuple. Furthermore, the dimension of $(\Omega, d)$, provided that it is flat, is the largest number $N$ such  that there exists a tuple of size $N+1$ with $|\boldsymbol V|>0$.
\end{definition}
\begin{thm} \label{thm_Morgan}
\textnormal{\citep{morgan1974embedding}} A metric space can be embedded into an $n$ dimensional Euclidean space if and only if the metric space is flat and of dimension less than or equal to $n$.
\end{thm}

\subsection*{Proof of Theorem \ref{thm1}}
\begin{proof}
We first prove that the metric space $(\Omega, u)$ has an isometric embedding into $n-1$ dimensional space by looking at each branch $k$ separately. For an arbitrary value of $k \in \{1,2, ..., K\}$, define $u_k(i,j) = \sqrt{b_k|p_{i,k} - p_{j,k}|}$. Obviously $(\Omega, u_k)$ is also a metric space. We shall prove that $(\Omega, u_k)$ has an isometric embedding into the Euclidean space. Using Theorem \ref{thm_Morgan}, we need to show the following two conditions are met:


\begin{enumerate}
\item \textit{Flatness:} Take an arbitrary ordered tuple with size $N \leq n$ from $(\Omega, u_k)$. Without loss of generality, we assume that the tuple consists of the first $N$ samples in $\Omega$, i.e. $(\omega_1, \omega_2, ..., \omega_{N})$. This means that the $i$th sample in the tuple has $p_{i,k}$ as its taxa proportion descending from branch $k$. According to Definition \ref{def_Morgan}, $\boldsymbol V$ is defined as 
\begin{equation} \label{Eij}
\boldsymbol V_{i,j} = b_k \big( |p_{i+1,k} - p_{1,k}| + |p_{j+1,k} - p_{1,k}| - |p_{i+1,k} - p_{j+1,k}| \big) / 2
\end{equation}

For flatness we need to show $|\boldsymbol V| \geq 0$. There are three possibilities on $p_{i,k}$'s:
\begin{enumerate}
\item If there exists $i$ such that $p_{i+1,k} = p_{1,k}$, then $\boldsymbol V_{i,j} = 0$ for all $j \Rightarrow |\boldsymbol V| = 0$.

\item If there exists $i$ and $j$ such that $p_{i+1,k} = p_{j+1,k}$, then the the $i$th and $j$th row of $\boldsymbol V$ are identical, leading to $|\boldsymbol V| = 0$.

\item If neither of the above is true, define a bijective sorting function $\tau: \{1,2,...,N-1 \} \rightarrow \{1,2,...,N-1 \}$ such that $p_{\tau(1)+1,k} < p_{\tau(2)+1,k} < ... < p_{\tau(N-1)+1,k}$. Furthermore, let $t = |\{p_{i+1,k}: p_{i+1,k} < p_{1,k} \text{ and } 1 \leq i \leq N-1 \}|$.

Let $\tilde{\boldsymbol V}$ be the matrix such that $\tilde{\boldsymbol V}_{i,j} = \boldsymbol V_{\tau(i),\tau(j)}$. Obviously $|\tilde{\boldsymbol V}| = |\boldsymbol V|$ and $\tilde{\boldsymbol V}$ is symmetric. Using (\ref{Eij}) and the definition of $\tau$, we see that the upper triangle of $\tilde{\boldsymbol V}$ satisfies the following properties:
\begin{enumerate}
\item If $i=j$, then $\tilde{\boldsymbol V}_{i,j} = b_k |p_{\tau(i)+1,k} - p_{1,k}|$
\item If $i<j\leq t$, then  $p_{\tau(i)+1,k}-p_{1,k} < p_{\tau(j)+1,k}-p_{1,k} < 0 \Rightarrow \tilde{\boldsymbol V}_{i,j} = b_k |p_{\tau(j)+1,k} - p_{1,k}|$
\item If $i \leq t < j$, then $(p_{\tau(i)+1,k}-p_{1,k})(p_{\tau(j)+1,k}-p_{1,k}) < 0 \Rightarrow \tilde{\boldsymbol V}_{i,j} = 0$.
\item If $t<i<j$, then $0 < p_{\tau(i)+1,k}-p_{1,k} < p_{\tau(j)+1,k}-p_{1,k} \Rightarrow \tilde{\boldsymbol V}_{i,j} = b_k |p_{\tau(i)+1,k} - p_{1,k}|$
\end{enumerate}

Combining the above properties of $\tilde{\boldsymbol V}$, we can write it in block form:
$$\tilde{\boldsymbol V} = \begin{pmatrix}
\tilde{\boldsymbol V}_1 & \mathbf 0 \\
\mathbf 0 & \tilde{\boldsymbol V}_2
\end{pmatrix}$$
where $\tilde{\boldsymbol V}_1 \in \mathbb R^{t \times t}$ and $\tilde{\boldsymbol V}_2 \in \mathbb R^{(N-1-t) \times (N-1-t)}$. According to Lemma \ref{lemma1}, $|\tilde{\boldsymbol V}_1| > 0$ and $|\tilde{\boldsymbol V}_2| > 0$. Therefore, $|\tilde{\boldsymbol V}| = |\boldsymbol V| > 0$
\end{enumerate}

\item \textit{Minimum dimension} The minimum dimension of such embedding is simply the largest $N$ such that $|\boldsymbol V| > 0$ for a certain tuple $(\kappa_0, \kappa_1, ..., \kappa_N)$ from $(\Omega, u_k)$. Notice that if all $p_{i,k}$'s are equal, then $N = 0$, leading to a trivial embedding into 0-dimensional space. 
\end{enumerate}

Now suppose $k^*$ satisfies that $p_{i,k^*}$ are all different for $1 \leq i \leq n$. According to the arguments above, $(\Omega, u_{k^*})$ has an isometric embedding into an Euclidean space. Furthermore, the minimum dimension of such embedding is $n-1$ since $|\boldsymbol V| > 0$ for the tuple $(\omega_1, \omega_2, ..., \omega_n)$ due to the argument in 1(c).

So far we have proven the existence of Euclidean embedding for each $(\Omega, u_k)$. Let $\boldsymbol \gamma_{k1}, ..., \boldsymbol \gamma_{kn}$ be the Euclidean vectors that embeds $(\Omega, u_k)$ with minimum dimension. For each $i$, we define $\boldsymbol\zeta_i$ by concatenating all $\boldsymbol \gamma_{ki}$ for $1\leq k \leq K$:
$$\boldsymbol \zeta_i = (\boldsymbol\gamma_{1i}', \boldsymbol\gamma_{2i}', ..., \boldsymbol\gamma_{Ki}')'$$
Since $u^2(i,j) = \sum_{k=1}^K u^2_k(i,j)$ for all $i$ and $j$, it follows that $(\boldsymbol\zeta_1, \boldsymbol\zeta_2, ..., \boldsymbol\zeta_n)$ would be the embedded Euclidean vectors that preserve the metric $u$. Furthermore, $\text{rank}(\boldsymbol\zeta_1, \boldsymbol\zeta_2, ..., \boldsymbol\zeta_n) \geq \text{rank} (\boldsymbol\gamma_{k^*1}, \boldsymbol\gamma_{k^*2}, ..., \boldsymbol\gamma_{k^*n}) = n-1$. It follows that the minimum dimension of $(\Omega, u)$'s embedding is $n-1$. Now choose $\boldsymbol \zeta_1$ as the origin so that embedded vector of $i$th element becomes $\tilde{\boldsymbol \zeta}_i = \boldsymbol\zeta_i - \boldsymbol\zeta_1$. Since $\text{rank}(\boldsymbol 0, \tilde{\boldsymbol \zeta}_2, \tilde{\boldsymbol \zeta}_3, ..., \tilde{\boldsymbol \zeta}_n) = n-1$, we can orthogonally project them onto $\mathbb R^{n-1}$, hence the existence of an Euclidean embedding with $n-1$ dimensions.

Let $\boldsymbol Q$ be an $n \times (n-1)$ matrix with $i$th row denoting the $n-1$ dimensional embedding of $i$th element in $(\Omega, u)$. Furthermore, assume each column of $\boldsymbol Q$ has mean zero, which has no impact on the Euclidean distance induced by $\boldsymbol Q$. The arguments provided in the previous paragraph shows that $\text{rank}(\boldsymbol Q) = n-1$. By definition, $\boldsymbol D_{i,j} = -\sum_{z=1}^{n-1} (\boldsymbol Q_{iz} - \boldsymbol Q_{jz})^2 / 2$, so 
\begin{align*}
\boldsymbol M = \boldsymbol{JDJ} = \boldsymbol{JQQ'J} = \boldsymbol{QQ'} 
\end{align*}
is positive semidefinite with rank $n-1$.

\end{proof}

\subsection*{Proof of Corollary \ref{cor1}}
\begin{proof}
Let $\boldsymbol Q$ be the same $n \times (n-1)$ matrix as defined above. For an arbitrary $\boldsymbol v \in \mathbb R^{n-m}$ and $\boldsymbol v \neq 0$, consider $\boldsymbol {v'LML'v} = \boldsymbol{(L'v)'M(L'v)}$. By definition of $\boldsymbol L$, we have $\boldsymbol 1_n \in ker(\boldsymbol L) = im(\boldsymbol L')^{\perp} \Rightarrow im(\boldsymbol L') \subset \mathbf 1^\perp$. 

Moreover, $\boldsymbol M \boldsymbol 1_n = \boldsymbol {QQ'} \boldsymbol 1_n = 0$ since $\boldsymbol Q$ is column-centered. Given that $\text{rank}(\boldsymbol M) = n-1$ from Theorem \ref{thm1}, it follows that $\boldsymbol 1_n$ is the only eigenvector of $\boldsymbol M$ corresponding to zero eigenvalue.

Combining the above two observations, we see that $im(\boldsymbol L')$ is a subspace of the space spanned by all eigenvectors of $\boldsymbol M$ that correspond to positive eigenvalues. Therefore, we have $\boldsymbol{(L'v)'M(L'v)} > 0 \Rightarrow \boldsymbol{LML'}$ is positive definite.

\end{proof}

\end{document}